\documentclass [11pt]{article}
\usepackage{stmaryrd}
\usepackage{}
\usepackage{amsmath,amsfonts,amssymb,mathbbol,amsthm,latexsym,bm,exscale,relsize,cite}
\usepackage{mathrsfs}
\usepackage{}
\usepackage{geometry}
\usepackage{amsmath,amsfonts, amsthm,amssymb,oldgerm}
\usepackage{graphicx}
\usepackage{array}
\usepackage{epsfig}
\usepackage[pdftex,bookmarksopen=true,colorlinks=true, linkcolor=red, citecolor=blue]{hyperref}
\usepackage{relsize}
\usepackage{indentfirst}
\usepackage{color}
\usepackage{indentfirst}
\renewcommand{\theequation}{\thesection.\arabic{equation}}
\begin{document}

\renewcommand{\theequation}{\arabic{section}.\arabic{equation}}
\def\t{mathbb{T}}
\def\r{\mathbb{R}}
\def\H{\mathcal{H}}
\def\x{\mathcal{X}}
\def\y{\mathcal{Y}}
\def\C{\mathbb{C}}
\def\n{\mathbb{N}}
\def\Z{\mathbb{Z}}
\def\<{\langle}
\def\>{\rangle}
\def\D{\boldsymbol{\mathcal{D}}}
\def\F{\mathcal{F}}
\def\2{L^2}
\def\c{L^2([0,T];H^1(0,1))}
\def\i{\infty}
\def\h{\hat}
\def\v{\vec}
\def\lam{\lambda}
\def\ds{\displaystyle}
\newtheorem{thm}{\bf Theorem}[section]
\newtheorem{defn}[thm]{\bf Definition}
\newtheorem{note}[thm]{\bf Notations}
\newtheorem{lem}[thm]{\bf Lemma}
\newtheorem{cor}[thm]{\bf Corollary}
\newtheorem{rem}[thm]{\bf Remark}
\newtheorem{open}[thm]{Open Problem}
\newtheorem{rems}[thm]{\bf Remarks}
\newtheorem{example}[thm]{\bf{Example}}
\newtheorem{assumption}[thm]{\bf{Assumption}}
\newtheorem{prop}[thm]{\bf{Proposition}}
\newcommand{\justkdv}[1]{  #1 _t  +#1_x +#1  #1 _x +  #1 _{xxx}=0}
\newcommand{\kdv}[1]{  #1 _t  +#1_x +#1  #1 _x +  #1 _{xxx}=0,\quad u(x,0)=\phi(x) }
\newcommand{\dt}{\psi (\delta ^{-1} t)}
\newcommand{\dxi}{(\xi ^3 -4\xi )}
\newcommand{\taoxi}{\left (1+|\tau -\dxi |\right )}
\newcommand{\bfactor}{\left <i(\tau -\dxi )+3\xi ^2\right >^{2b}}
\newcommand{\zetafactor}{\left <\tau -\zeta \right >^{2b}}
\newcommand{\sfactor}{\left <\xi \right >^{2s}}
\newcommand{\rinfty}{\int ^{\infty}_{-\infty}}
\renewcommand{\abstractname}{}
\renewcommand{\theequation}{\thesection.\arabic{equation}}
\numberwithin{equation}{section}
\theoremstyle{definition}
\renewcommand{\proofname}{\noindent Proof}
\newtheorem{Definition}{\noindent Definition}[section]
\renewcommand{\refname}{\center{References}}

\title{\Large Cyclic Lattices, Ideal Lattices and Bounds for the Smoothing Parameter}

\author{Zhiyong Zheng$^{~a}$ ~~ Fengxia Liu$^{~b,*}$ ~~ Yunfan Lu$^{~c}$ ~~ Kun Tian$~^{~d}$ \\[10pt]
\small {$^{a,b,c,d }$  Engineering Research Center of Ministry of Education for Financial Computing} \\
\small {and Digital Engineering, Renmin University of China,}\\
\small { Beijing, 100872, P. R. China}\\
[5pt]}

\footnotetext{*Corresponding author (Fengxia Liu).\\
E-mail addresses:  liu$\_$fx@ruc.edu.cn  (F. Liu).}

\date{}
\maketitle

\noindent \textbf{Abstract:} Cyclic lattices and ideal lattices were introduced by Micciancio in \cite{D2}, Lyubashevsky and Micciancio in \cite{L1} respectively, which play an efficient role in Ajtai's construction of a collision resistant Hash function (see \cite{M1} and \cite{M2}) and in Gentry's construction of fully homomorphic encryption (see \cite{G}). Let $R=Z[x]/\langle \phi(x)\rangle$ be a quotient ring of the integer coefficients polynomials ring, Lyubashevsky and Micciancio regarded an ideal lattice as the correspondence of an ideal of $R$, but they neither explain how to extend this definition to whole Euclidean space $\mathbb{R}^n$, nor exhibit the relationship of cyclic lattices and ideal lattices.

In this paper, we regard the cyclic lattices and ideal lattices as the correspondences of finitely generated $R$-modules, so that we may show that ideal lattices are actually a special subclass of cyclic lattices, namely, cyclic integer lattices. In fact, there is a one to one correspondence between cyclic lattices in $\mathbb{R}^n$ and finitely generated $R$-modules (see Theorem \ref{th4} below). On the other hand, since $R$ is a Noether ring, each ideal of $R$ is a finitely generated $R$-module, so it is natural and reasonable to regard ideal lattices as a special subclass of cyclic lattices (see corollary \ref{co3.4} below). It is worth noting that we use more general rotation matrix here, so our definition and results on cyclic lattices and ideal lattices are more general forms. As application, we provide cyclic lattice with an explicit and countable upper bound for the smoothing parameter (see Theorem \ref{th5} below). It is an open problem that is the shortest vector problem on cyclic lattice NP-hard? (see \cite{D2}). Our results may be viewed as a substantial progress in this direction.

\vspace{0.4cm}
\noindent \textbf{Keywords:} Cyclic Lattice, Ideal Lattice, Finitely Generated $R$-module, Smoothing Parameter.

\numberwithin{equation}{section}
\section{Discrete Subgroup in $\mathbb{R}^n$}

Let $\mathbb{R}$ be the real numbers field, $\mathbb{Z}$ be the integers ring and $\mathbb{R}^n$ be Euclidean space of which is an $n$-dimensional linear space over $\mathbb{R}$ with the Euclidean norm $|x|$ given by
\begin{equation*}
|x|=(\sum_{i=1}^{n} x_i^2)^{\frac{1}{2}},\quad \text{where}\ x'=( x_1,x_1,  \cdots,  x_{n})
     \in \mathbb{R}^n.
\end{equation*}
We use column vector notation for $\mathbb{R}^n$ through out this paper, and $x'=(x_1,x_2,\dots,x_n)$ is transpose of $x$, which is called row vector of $\mathbb{R}^n$.\\

\noindent \textbf{Definition 1.1} Let $L\subset \mathbb{R}^n$ be a non-trivial additive subgroup, it is called a discrete subgroup if there is a positive real number $\lambda>0$ such that
\begin{equation}
\min\limits_{x\in L,x\neq 0} |x|\geqslant \lambda>0.
\end{equation}
As usual, a ball of center $x_0$ with radius $\delta$ is defined by
\begin{equation*}
b(x_0,\delta)=\{x\in \mathbb{R}^n\ \Big|\ |x-x_0|\leqslant \delta\}.
\end{equation*}
If $L$ is a discrete subgroup of $\mathbb{R}^n$, then there are only finitely many vectors of $L$ lie in every ball $b(0,\delta)$, thus we always find a vector $\alpha\in L$ such that
\begin{equation}
|\alpha|=\min\limits_{x\in L,x\neq 0} |x|=\lambda>0,\quad \alpha\in L.
\end{equation}
$\alpha$ is called one of shortest vector of $L$ and $\lambda$ is called the minimum distance of $L$.

Let $B=[\beta_1,\beta_2,\dots,\beta_m]\in \mathbb{R}^{n\times m}$ be a $n\times m$ dimensional matrix with rank$(B)=m\leqslant n$, it means that $\beta_1,\beta_2,\dots,\beta_m$ are $m$ linearly independent vectors in $\mathbb{R}^n$. The lattice $L(B)$ generated by $B$ is defined by
\begin{equation}
L(B)=\sum_{i=1}^{m} x_i \beta_i =\{Bx\ |\ x\in \mathbb{Z}^m\},\quad \forall x_i\in \mathbb{Z}.
\end{equation}
which is all linear combinations of $\beta_1,\beta_2,\dots,\beta_m$ over $\mathbb{Z}$. If $m=n$, $L(B)$ is called a full-rank lattice.

It is a well-known conclusion that a discrete subgroup $L$ in $\mathbb{R}^n$ is just a lattice $L(B)$. Firstly, we give a detail proof here by making use of the simultaneous Diophantine approximation theory in real number field $\mathbb{R}$ (see \cite{J1} and \cite{J2}).\\

\begin{lem}\label{lm1.1} Let $L\subset \mathbb{R}^n$ be a discrete subgroup, $\alpha_1,\alpha_2,\dots,\alpha_m \in L$ be $m$ vectors of $L$. Then $\alpha_1,\alpha_2,\dots,\alpha_m$ are linearly independent over $\mathbb{R}$, if and only if which are linearly independent over $\mathbb{Z}$.
\end{lem}
\begin{proof} If $\alpha_1,\alpha_2,\dots,\alpha_m$ are linearly independent over $\mathbb{R}$, trivially which are linearly independent over $\mathbb{Z}$. Suppose that $\alpha_1,\alpha_2,\dots,\alpha_m$ are linearly independent over $\mathbb{Z}$, we consider arbitrary linear combination over $\mathbb{R}$. Let
\begin{equation}
a_1 \alpha_1+a_2 \alpha_2+\cdots+a_m \alpha_m=0,\quad \forall a_i\in \mathbb{R}.
\end{equation}
We should prove (1.4) is equivalent to $a_1=a_2=\cdots=a_m=0$, which implies that $\alpha_1,\alpha_2,\dots,\alpha_m$ are linearly independent over $\mathbb{R}$.

By Minkowski's Third Theorem (see Theorem  VII of \cite{J2}), for any sufficiently large $N>1$, there are a positive integer $q\geqslant 1$ and integers $p_1,p_2,\dots,p_m \in \mathbb{Z}$ such that
\begin{equation}
\max\limits_{1\leqslant i\leqslant m} |qa_i-p_i|<N^{-\frac{1}{m}},\ \text{and}\ 1\leqslant q\leqslant N.
\end{equation}
By (1.4), we have
\begin{equation*}
|p_1 \alpha_1+p_2 \alpha_2+\cdots+p_m \alpha_m|=|(qa_1-p_1)\alpha_1+(qa_2-p_2)\alpha_2+\cdots+(qa_m-p_m)\alpha_m|
\end{equation*}
\begin{equation}
\leqslant mN^{-\frac{1}{m}} \max\limits_{1\leqslant i\leqslant m}|\alpha_i|.\qquad
\end{equation}
Let $\lambda$ be the minimum distance of $L$, $\epsilon>0$ be any positive real number. We select $N$ such that
\begin{equation*}
N>\max \{(\frac{m}{\epsilon})^m,\ (\frac{m}{\lambda})^m \max\limits_{1\leqslant i\leqslant m} |\alpha_i|^m\}.
\end{equation*}
It follows that $mN^{-\frac{1}{m}}<\epsilon$ and
\begin{equation*}
mN^{-\frac{1}{m}} \max\limits_{1\leqslant i\leqslant m} |\alpha_i|<\lambda.
\end{equation*}
By (1.6) we have
\begin{equation*}
|p_1 \alpha_1+p_2 \alpha_2+\cdots+p_m \alpha_m|<\lambda.
\end{equation*}
Since $p_1 \alpha_1+p_2 \alpha_2+\cdots+p_m \alpha_m \in L$, thus we have $p_1 \alpha_1+p_2 \alpha_2+\cdots+p_m \alpha_m=0$, and $p_1=p_2=\cdots=p_m=0$. By (1.5) we have $q|a_i|<\frac{1}{m} \varepsilon$ for all $i$, $1\leqslant i\leqslant m$. Since $\varepsilon$ is sufficiently small positive number, we must have $a_1=a_2=\cdots=a_m=0$. We complete the proof of lemma.

\end{proof}

Suppose that $B\in \mathbb{R}^{n\times m}$ is an $n\times m$-dimensional matrix and rank$(B)=m$, $B'$ is the transpose of $B$. It is easy to verify
\begin{equation*}
\text{rank}(B'B)=\text{rank}(B)=m\Rightarrow\ \text{det}(B'B)\neq 0,
\end{equation*}
which implies that $B'B$ is an invertible square matrix of $m\times m$ dimension. Since $B'B$ is a positive defined symmetric matrix, then there is an orthogonal matrix $P\in \mathbb{R}^{m\times m}$ such that
\begin{equation}
P'B'BP=\text{diag}\{\delta_1,\delta_2,\dots,\delta_m\},
\end{equation}
where $\delta_i>0$ are the characteristic value of $B'B$, and diag$\{\delta_1,\delta_2,\dots,\delta_m\}$ is the diagonal matrix of $m\times m$ dimension.\\

\begin{lem}\label{lm1.2} Suppose that $B\in \mathbb{R}^{n\times m}$ with rank$(B)=m$, $\delta_1,\delta_2,\dots,\delta_m$ are $m$ characteristic values of $B'B$, and $\lambda(L(B))$ is the minimum distance of lattice $L(B)$, then we have
\begin{equation}
\lambda(L(B))=\min\limits_{x\in \mathbb{Z}^m,\ x\neq 0} |Bx|\geqslant \sqrt{\delta},
\end{equation}
where $\delta=\min\{\delta_1,\delta_2,\dots,\delta_m\}$.
\end{lem}
\begin{proof} Let $A=B'B$, by (1.7), there exists an orthogonal matrix $P\in \mathbb{R}^{m\times m}$ such that
\begin{equation*}
P'AP=\text{diag}\{\delta_1,\delta_2,\dots,\delta_m\}.
\end{equation*}
If $x\in \mathbb{Z}^m$, $x\neq 0$, we have
\begin{equation*}
|Bx|^2=x'Ax=x'P(P'AP)P'x
\end{equation*}
\begin{equation*}
\qquad\qquad\quad\ \ =(P'x)'\ \text{diag}\{\delta_1,\delta_2,\dots,\delta_m\}P'x
\end{equation*}
\begin{equation*}
\geqslant \delta |P'x|^2=\delta |x|^2.\ \
\end{equation*}
Since $x\in \mathbb{Z}^m$ and $x\neq 0$, we have $|x|^2\geqslant 1$, it follows that
\begin{equation*}
\min\limits_{x\in \mathbb{Z}^m,\ x\neq 0} |Bx|\geqslant \sqrt{\delta} |x| \geqslant \sqrt{\delta}.
\end{equation*}
We have lemma \ref{lm1.2} immediately.

\end{proof}

Another application of Lemma \ref{lm1.2} is to give a countable upper bound for smoothing parameter (see Theorem \ref{th5} below).
Combining lemma \ref{lm1.1} and lemma \ref{lm1.2}, we show that the following assertion.\\

\begin{thm}\label{th1} Let $L\subset \mathbb{R}^n$ be a subset, then $L$ is a discrete subgroup if and only if there is an $n\times m$ dimensional matrix $B\in \mathbb{R}^{n\times m}$ with rank$(B)=m$ such that
\begin{equation}
L=L(B)=\{Bx\ |\ x\in \mathbb{Z}^m\}.
\end{equation}
\end{thm}

\begin{proof} If $L \subset \mathbb{R}^n$ is a discrete subgroup, then $L$ is a free $\mathbb{Z}$-module. By lemma \ref{lm1.1}, we have $\text{rank}_{\mathbb{Z}}(L)=m\leqslant n$. Let $\beta_1,\beta_2,\dots,\beta_m$ be a $\mathbb{Z}$-basis of $L$, then
\begin{equation*}
L=\{\sum_{i=1}^{m} a_i \beta_i\ |\ a_i\in \mathbb{Z}\}.
\end{equation*}
Writing $B=[\beta_1,\beta_2,\dots,\beta_m]_{n\times m}$, then the rank of matrix $B$ is $m$, and
\begin{equation*}
L=\{Bx\ |\ x\in \mathbb{Z}^m\}=L(B).
\end{equation*}
Conversely, let $L(B)$ be arbitrary lattice generated by $B$, obviously, $L(B)$ is an additive subgroup of $\mathbb{R}^n$, by lemma \ref{lm1.2}, $L(B)$ is also a discrete subgroup, we have Theorem \ref{th1} at once.

\end{proof}

\begin{cor}\label{co1.1} Let $L\subset \mathbb{R}^n$ be a lattice and $G\subset L$ be an additive subgroup of $L$, then $G$ is a lattice of $\mathbb{R}^n$.
\end{cor}
\begin{cor}\label{co1.2} Let $L\subset \mathbb{Z}^n$ be an additive subgroup, then $L$ is a lattice of $\mathbb{R}^n$. These lattices are called integer lattices.
\end{cor}
According to above Theorem \ref{th1}, a lattice $L(B)$ is equivalent to a discrete subgroup of $\mathbb{R}^n$. Suppose $L=L(B)$ is a lattice with generated matrix $B\in \mathbb{R}^{n\times m}$, and rank$(B)=m$, we write rank$(L)=$rank$(B)$, and
\begin{equation}
d(L)=\sqrt{\text{det}(B'B)}.
\end{equation}
In particular, if rank$(L)=n$ is a full-rank lattice, then $d(L)=|\text{det}(B)|$ as usual. A sublattice $N$ of $L$ means a discrete additive subgroup of $L$, the quotient group is written by $L/N$ and the cardinality of $L/N$ is denoted by $|L/N|$.\\

 \begin{lem}\label{lm1.3}Let $L\subset \mathbb{R}^n$ be a lattice and $N\subset L$ be a sublattice. If rank$(N)=$rank$(L)$, then the quotient group $L/N$ is a finite group.
\end{lem}
\begin{proof}Let rank$(L)=m$, and $L=L(B)$, where $B\in \mathbb{R}^{n\times m}$ with rank$(B)=m$. We define a mapping $\sigma$ from $L$ to $\mathbb{Z}^m$ by $\sigma(Bx)=x$. Clearly, $\sigma$ is an additive group isomorphism, $\sigma(N)\subset \mathbb{Z}^m$ is a full-rank lattice of $\mathbb{Z}^m$, and $L/N \cong \mathbb{Z}^m/\sigma(N)$. It is a well-known result that
\begin{equation*}
|\mathbb{Z}^m/\sigma(N)|=d(\sigma(N)).
\end{equation*}
It follows that
\begin{equation*}
|L/N|=|\mathbb{Z}^m/\sigma(N)|=d(\sigma(N)).
\end{equation*}
Lemma \ref{lm1.3} follows.

\end{proof}

Suppose that $L_1\subset \mathbb{R}^n$, $L_2\subset \mathbb{R}^n$ are two lattices of $\mathbb{R}^n$, we define $L_1+L_2=\{a+b|a\in L_1,b\in L_2\}$. Obviously, $L_1+L_2$ is an additive subgroup of $\mathbb{R}^n$, but generally speaking, $L_1+L_2$ is not a lattice of $\mathbb{R}^n$ again.\\

\begin{lem}\label{lm1.4}  Let $L_1\subset \mathbb{R}^n$, $L_2\subset \mathbb{R}^n$ be two lattices of $\mathbb{R}^n$. If rank$(L_1 \cap L_2)=$rank$(L_1)$ or rank$(L_1 \cap L_2)=$rank$(L_2)$, then $L_1+L_2$ is again a lattice of $\mathbb{R}^n$.
\end{lem}
\begin{proof} To prove $L_1+L_2$ is a lattice of $\mathbb{R}^n$, by Theorem \ref{th1}, it is sufficient to prove $L_1+L_2$ is a discrete subgroup of $\mathbb{R}^n$. Suppose that rank$(L_1 \cap L_2)=$rank$(L_1)$, for any $x\in L_1$, we define a distance function $\rho(x)$ by
\begin{equation*}
\rho(x)=\inf \{|x-y|\ \Big|\ y\neq x,\ y\in L_2\}.
\end{equation*}
Since there are only finitely many vectors in $L_2\cap b(x,\delta)$, where $b(x,\delta)$ is any a ball of center $x$ with radius $\delta$. Therefore, we have
\begin{equation}
\rho(x)=\min \{|x-y|\ \Big|\ y\neq x,\ y\in L_2\}=\lambda_x>0.
\end{equation}
On the other hand, if $x_1\in L_1$, $x_2\in L_1$ and $x_1-x_2\in L_2$, then there is $y_0\in L_2$ such that $x_1=x_2+y_0$, and we have $\rho(x_1)=\rho(x_2)$. It means that $\rho(x)$ is defined over the quotient group $L_1+L_2/L_2$. Because we have the following group isomorphic theorem
\begin{equation*}
L_1+L_2/L_2\cong L_1/L_1\cap L_2.
\end{equation*}
By lemma \ref{lm1.3}, it follows that
\begin{equation*}
|L_1+L_2/L_2|=|L_1/L_1\cap L_2|<\infty.
\end{equation*}
In other words, $L_1+L_2/L_2$ is also a finite group. Let $x_1,x_2,\dots,x_k$ be the representative elements of $L_1+L_2/L_2$, we have
\begin{equation*}
\min\limits_{x\in L_1,y\in L_2,x\neq y} |x-y|=\min\limits_{1\leqslant i\leqslant k} \rho(x_i)\geqslant \min\{\lambda_{x_1},\lambda_{x_2},\dots,\lambda_{x_k}\}>0.
\end{equation*}
Therefore, $L_1+L_2$ is a discrete subgroup of $\mathbb{R}^n$, thus it is a lattice of $\mathbb{R}^n$ by Theorem \ref{th1}.

\end{proof}

\begin{rem} \label{rm1} The condition rank$(L_1 \cap L_2)=$rank$(L_1)$ or rank$(L_1 \cap L_2)=$rank$(L_2)$ in lemma \ref{lm1.4} seems to be necessary. As a counterexample, we see the real line $\mathbb{R}$, let $L_1=\mathbb{Z}$ and $L_2=\sqrt{2}\mathbb{Z}$, then $L_1+L_2$ is not a discrete subgroup of $\mathbb{R}$, thus $L_1+L_2$ is not a lattice in $\mathbb{R}$. Because $L_1+L_2=\{n+\sqrt{2}m\big| n\in \mathbb{Z},m\in \mathbb{Z}\}$ is dense in $\mathbb{R}$ by Dirichlet's Theorem (see Theorem I of [5]).
\end{rem}
As a direct consequence, we have the following generalized form of lemma \ref{lm1.4}.\\

\begin{cor} \label{co1.3}  Let $L_1,L_2,\dots,L_m$ be $m$ lattices of $\mathbb{R}^n$ and
\begin{equation*}
\text{rank}(L_{1}\cap L_{2}\cap\cdots\cap L_{m})=\text{rank}(L_{j})\ \text{for some}\ 1\leqslant j\leqslant m.
\end{equation*}
Then $L_1+L_2+\cdots+L_m$ is a lattice of $\mathbb{R}^n$.
\end{cor}
\begin{proof}  Without loss of generality, we assume that
\begin{equation*}
\text{rank}(L_{1}\cap L_{2}\cap\cdots\cap L_{m})=\text{rank}(L_m).
\end{equation*}
Let $L_1+L_2+\cdots+L_{m-1}=L'$, then
\begin{equation*}
L'+L_m/L'\cong L_m/L'\cap L_m.
\end{equation*}
Since rank$(L'\cap L_m)=$rank$(L_m)$, by lemma \ref{lm1.4}, we have $L'+L_m=L_1+L_2+\cdots+L_m$ is a lattice of $\mathbb{R}^n$ and the corollary follows.

\end{proof}

\section{Ideal Matrices}

Let $\mathbb{R}[x]$ and $\mathbb{Z}[x]$ be the polynomials rings over $\mathbb{R}$ and $\mathbb{Z}$ with variable $x$ respectively. Suppose that
\begin{equation}
\phi(x)=x^n-\phi_{n-1}x^{n-1}-\cdots-\phi_1x-\phi_0\in \mathbb{Z}[x],\ \phi_0\neq 0,
\end{equation}
is a polynomial with integer coefficients of which has no multiple roots in complex numbers field $\mathbb{C}$. Let $w_1,w_2,\dots,w_n$ be the $n$ different roots of $\phi(x)$ in $\mathbb{C}$, the Vandermonde matrix $V_{\phi}$ is defined by
\begin{equation}
V_{\phi}=
\begin{pmatrix}
1 & 1 & \cdots & 1 \\
w_1 & w_2 & \cdots & w_n \\
\vdots & \vdots &  & \vdots \\
w_1^{n-1} & w_2^{n-1} & \cdots & w_n^{n-1}
\end{pmatrix},\quad
\text{\ and\ \ \ det}(V_{\phi})\neq 0.
\end{equation}
According to the given polynomial $\phi(x)$, we define a rotation matrix $H=H_{\phi}$ by
\begin{equation}
H=H_{\phi}=
\left(
\begin{array}{ccc|c}
0 & \cdots & 0 & \phi_0\\
\hline
& & & \phi_1\\
& I_{n-1} & & \vdots \\
& & & \phi_{n-1} \\
\end{array}
\right)_{n\times n}\in \mathbb{Z}^{n\times n},
\end{equation}
where $I_{n-1}$ is the $(n-1)\times (n-1)$ unit matrix. Obviously, the characteristic polynomial of $H$ is just $\phi(x)$.

We use column notation for vectors in $\mathbb{R}^n$, for any $f=\begin{pmatrix}
     f_0 \\
     f_1 \\
     \vdots \\
     f_{n-1}
\end{pmatrix}\in \mathbb{R}^n$, the ideal matrix generated by vector $f$ is defined by
\begin{equation}
H^*(f)=[f,Hf,H^2 f,\dots,H^{n-1}f]_{n\times n}\in \mathbb{R}^{n\times n},
\end{equation}
which is a block matrix in terms of each column $H^k f\ (0\leqslant k\leqslant n-1)$. Sometimes, $f$ is called an input vector. It is easily seen that $H^*(f)$ is a more general form of the classical circulant matrix (see [6]) and $r$-circulant matrix (see \cite{b} and \cite{y}). In fact, if $\phi(x)=x^n-1$, then $H^*(f)$ is the ordinary circulant matrix generated by $f$. If $\phi(x)=x^n-r$, then $H^*(f)$ is the $r$-circulant matrix.

By (2.4), it follows immediately that
\begin{equation}
H^*(f+g)=H^*(f)+H^*(g),\ \text{and}\ H^*(\lambda f)=\lambda H^*(f),\ \forall \lambda\in \mathbb{R}.
\end{equation}
Moreover, $H^*(f)=0$ is a zero matrix if and only if $f=0$ is a zero vector, thus one has $H^*(f)=H^*(g)$ if and only if $f=g$. Let $M^*$ be the set of all ideal matrices, namely
\begin{equation}
M^*=\{H^*(f)\ |\ f\in \mathbb{R}^n\}.
\end{equation}
We may regard $H^*$ as a mapping from $\mathbb{R}^n$ to $M^*$ of which is a one to one correspondence.

In \cite{z}, we have shown that some basic properties for ideal matrix, most of them may be summarized as the following theorem.\\

\begin{thm}\label{th2} Suppose that $\phi(x)\in \mathbb{Z}[x]$ is a fixed polynomial with no multiple roots in $\mathbb{C}$, then for any two column vectors $f$ and $g$ in $\mathbb{R}^n$, we have

(i) $H^*(f)=f_0 I_n+f_1 H+\cdots+f_{n-1}H^{n-1}$;

(ii) $H^*(f)H^*(g)=H^*(H^*(f)g)$ and $H^*(f)H^*(g)=H^*(g)H^*(f)$;

(iii) $H^*(f)=V_{\phi}^{-1}\ \text{diag}\{f(w_1),f(w_2),\dots,f(w_n)\}V_{\phi}$;

(iv) det $(H^*(f))=\Pi_{i=1}^n f(w_i)$;

(v) $H^*(f)$ is an invertible matrix if and only if $(f(x),\phi(x))=1$ in $\mathbb{R}[x]$.

\noindent where $V_{\phi}$ is the Vandermonde matrix given by (2.2), $w_i\ (1\leqslant i\leqslant n)$ are all roots of $\phi(x)$ in $\mathbb{C}$, and diag$\{f(w_1),f(w_2),\dots,f(w_n)\}$ is the diagonal matrix.
\end{thm}

\begin{proof}
 See Theorem 2 of \cite{z}.
\end{proof}

Let $e_1,e_2,\dots,e_n$ be unit vectors of $\mathbb{R}^n$, that is
\begin{equation*}
e_1=
\begin{pmatrix}
     1 \\
     0 \\
     \vdots \\
     0
\end{pmatrix},
e_2=
\begin{pmatrix}
     0 \\
     1 \\
     \vdots \\
     0
\end{pmatrix},\cdots,
e_n=
\begin{pmatrix}
     0 \\
     0 \\
     \vdots \\
     1
\end{pmatrix}.
\end{equation*}
It is easy to verify that
\begin{equation}
H^*(e_1)=I_n,\ \text{and}\ H^*(e_k)=H^{k-1},\ 1\leqslant k\leqslant n.
\end{equation}
This means that the unit matrix $I_n$ and rotation matrices $H^k\ (1\leqslant k\leqslant n-1)$ are all the ideal matrices.

Let $\phi(x)\mathbb{R}[x]$ and $\phi(x)\mathbb{Z}[x]$ be the principal ideals generated by $\phi(x)$ in $\mathbb{R}[x]$ and $\mathbb{Z}[x]$ respectively, we denote the quotient rings $R$ and $\overline{R}$ by
\begin{equation}
R=\mathbb{Z}[x]/\phi(x)\mathbb{Z}[x],\ \text{and}\ \overline{R}=\mathbb{R}[x]/\phi(x)\mathbb{R}[x].
\end{equation}
There is a one to one correspondence between $\overline{R}$ and $\mathbb{R}^n$ given by
\begin{equation*}
f(x)=f_0+f_1 x+\cdots+f_{n-1}x^{n-1}\in \overline{R} \xrightarrow{\quad t\quad} f=\begin{pmatrix}
     f_0 \\
     f_1 \\
     \vdots \\
     f_{n-1}
\end{pmatrix}\in \mathbb{R}^n.
\end{equation*}
We denote this correspondence by $t$, that is
\begin{equation}
t(f(x))=f\ \text{and}\ t^{-1}(f)=f(x),\ \forall f(x)\in \overline{R},\ \text{and}\ f\in \mathbb{R}^n.
\end{equation}
If we restrict $t$ in the quotient ring $R$, then which gives a one to one correspondence between $R$ and $\mathbb{Z}^n$. First, we show that $t$ is also a ring isomorphism.\\

\begin{defn} \label{def2.1}For any two column vectors $f$ and $g$ in $\mathbb{R}^n$, we define the $\phi$-convolutional product $f*g$ by $f*g=H^*(f)g$.

By Theorem \ref{th2}, it is easy to see that
\begin{equation}
f*g=g*f,\ \text{and}\ H^*(f*g)=H^*(f)H^*(g).
\end{equation}
\end{defn}

\begin{lem}\label{lm2.1} For any two polynomials $f(x)$ and $g(x)$ in $\overline{R}$, we have
\begin{equation*}
t(f(x)g(x))=H^*(f)g=f*g.
\end{equation*}
\end{lem}

\begin{proof} Let $g(x)=g_0+g_1 x+\cdots+g_{n-1}x^{n-1}\in \overline{R}$, then
\begin{equation*}
xg(x)=\phi_0 g_{n-1}+(g_0+\phi_1 g_{n-1})x+\cdots+(g_{n-2}+\phi_{n-1}g_{n-1})x^{n-1}.
\end{equation*}
It follows that
\begin{equation}
t(xg(x))=Ht(g(x))=Hg.
\end{equation}
Hence, for any $0\leqslant k\leqslant n-1$, we have
\begin{equation}
t(x^k g(x))=H^k t(g(x))=H^k g,\ 0\leqslant k\leqslant n-1.
\end{equation}
Let $f(x)=f_0+f_1 x+\cdots+f_{n-1}x^{n-1}\in \overline{R}$, by (i) of Theorem \ref{th2}, we have
\begin{equation*}
t(f(x)g(x))=\sum_{i=0}^{n-1} f_i t(x^i g(x))=\sum_{i=0}^{n-1} f_i H^i g=H^*(f)g.
\end{equation*}
The lemma follows.

\end{proof}

\begin{thm}\label{tm3}Under $\phi$-convolutional product, $\mathbb{R}^n$ is a commutative ring with identity element $e_1$ and $\mathbb{Z}^n\subset \mathbb{R}^n$ is its subring. Moreover, we have the following ring isomorphisms
\begin{equation*}
\overline{R}\cong \mathbb{R}^n \cong M^*,\ \text{and}\ R\cong \mathbb{Z}^n\cong M_{\mathbb{Z}}^{*},
\end{equation*}
where $M^*$ is the set of all ideal matrices given by (2.6), and $M_{\mathbb{Z}}^{*}$ is the set of all integer ideal matrices.
\end{thm}

\begin{proof} Let $f(x)\in\overline{R}$ and $g(x)\in\overline{R}$, then
\begin{equation*}
t(f(x)+g(x))=f+g=t(f(x))+t(g(x)),
\end{equation*}
and
\begin{equation*}
t(f(x)g(x))=H^*(f)g=f*g=t(f(x))*t(g(x)).
\end{equation*}
This means that $t$ is a ring isomorphism. Since $f*g=g*f$ and $e_1*g=H^*(e_1)g=I_n g=g$, then $\mathbb{R}^n$ is a commutative ring with $e_1$ as the identity elements. Noting $H^*(f)$ is an integer matrix if and only if $f\in \mathbb{Z}^n$ is an integer vector, the isomorphism of subrings follows immediately.

\end{proof}

According to property (v) of Theorem \ref{th2}, $H^*(f)$ is an invertible matrix whenever $(f(x),\phi(x))=1$ in $\mathbb{R}[x]$, we show that the inverse of an ideal matrix is again an ideal matrix.\\

\begin{lem}\label{lm2.2} Let $f(x)\in \overline{R}$ and $(f(x),\phi(x))=1$ in $\mathbb{R}[x]$, then
\begin{equation*}
(H^*(f))^{-1}=H^*(u).
\end{equation*}
where $u(x)\in\overline{R}$ is the unique polynomial such that $u(x)f(x)\equiv 1$ (mod $\phi(x)$).
\end{lem}
\begin{proof}By lemma \ref{lm2.1}, we have $u*f=e_1$, it follows that
\begin{equation*}
H^*(u)H^*(f)=H^*(e_1)=I_n.
\end{equation*}
Thus we have $(H^*(f))^{-1}=H^*(u)$. It is worth to note that if $H^*(f)$ is an invertible integer matrix, then $(H^*(f))^{-1}$ is not an integer matrix in general.

\end{proof}

Sometimes, the following lemma may be useful, especially, when we consider an integer matrix.\\

\begin{lem}\label{lm2.3}  Let $f(x)\in \mathbb{Z}[x]$ and $(f(x),\phi(x))=1$ in $\mathbb{Z}[x]$, then we have $(f(x),\phi(x))=1$ in $\mathbb{R}[x]$.
\end{lem}
\begin{proof} Let $Q$ be the rational number field. Since $(f(x),\phi(x))=1$ in $\mathbb{Z}[x]$, then $(f(x),\phi(x))=1$ in $\mathbb{Q}[x]$. We know that $\mathbb{Q}[x]$ is a principal  ideal domain, thus there are two polynomials $a(x)$ and $b(x)$ in $\mathbb{Q}[x]$ such that
\begin{equation*}
a(x)f(x)+b(x)\phi(x)=1.
\end{equation*}
This means that $(f(x),\phi(x))=1$ in $\mathbb{R}[x]$.

\end{proof}

\section{Cyclic Lattices and Ideal Lattices}

As we known that cyclic code play a central role in algebraic coding theorem (see Chapter 6 of \cite{L}). In \cite{z}, we extended ordinary cyclic code to more general forms, namely $\phi$-cyclic codes. To obtain an analogous concept of $\phi$-cyclic code in $\mathbb{R}^n$, we note that every rotation matrix $H$ defines a linear transformation of $\mathbb{R}^n$ by $x\rightarrow Hx$.\\

\begin{defn}\label{def3.1} A linear subspace $C\subset \mathbb{R}^n$ is called a $\phi$-cyclic subspace if $\forall \alpha\in C\Rightarrow H\alpha\in C$. A lattice $L\subset \mathbb{R}^n$ is called a $\phi$-cyclic lattice if $\forall \alpha\in L\Rightarrow H\alpha\in L$.

\end{defn}
In other words, a $\phi$-cyclic subspace $C$ is a linear subspace of $\mathbb{R}^n$, of which is closed under linear transformation $H$. A $\phi$-cyclic lattice $L$ is a lattice of $\mathbb{R}^n$ of which is closed under $H$. If $\phi(x)=x^n-1$, then $H$ is the classical circulant matrix and the corresponding cyclic lattice was first appeared in Micciancio \cite{D2}, but he do not discuss the further property for these lattices. To obtain the explicit algebraic construction of $\phi$-cyclic lattice, we first show that there is a one to one correspondence between $\phi$-cyclic subspaces of $\mathbb{R}^n$ and the ideals of $\overline{R}$.\\

\begin{lem} \label{lm3.1} Let $t$ be the correspondence between $\overline{R}$ and $\mathbb{R}^n$ given by (2.9), then a subset $C\subset \mathbb{R}^n$ is a $\phi$-cyclic subspace of $\mathbb{R}^n$, if and only if $t^{-1}(C)\subset \overline{R}$ is an ideal.
\end{lem}

\begin{proof} We extend the correspondence $t$ to subsets of $\overline{R}$ and $\mathbb{R}^n$ by
\begin{equation}
C(x)\subset \overline{R} \xrightarrow{\quad t\quad} C=\{c|c(x)\in C(x)\}\subset \mathbb{R}^n.
\end{equation}
Let $C(x)\subset \overline{R}$ be an ideal, it is clear that $C\subset t(C(x))$ is a linear subspace of $\mathbb{R}^n$. To prove $C$ is a $\phi$-cyclic subspace, we note that if $c(x)\in C(x)$, then by (2.11)
\begin{equation*}
xc(x)\in C(x)\Leftrightarrow Ht(c(x))=Hc\in C.
\end{equation*}
Therefore, if $C(x)$ is an ideal of $\overline{R}$, then $t(C(x))=C$ is a $\phi$-cyclic subspace of $\mathbb{R}^n$. Conversely, if $C\subset \mathbb{R}^n$ is a $\phi$-cyclic subspace, then for any $k\geqslant 1$, we have $H^k c\in C$ whenever $c\in C$, it implies
\begin{equation*}
\forall c(x)\in C(x)\Rightarrow x^k c(x)\in C(x),\ 0\leqslant k\leqslant n-1,
\end{equation*}
which means that $C(x)$ is an ideal of $\overline{R}$. We complete the proof.

\end{proof}

By above lemma, to find a $\phi$-cyclic subspace in $\mathbb{R}^n$, it is enough to find an ideal of $\overline{R}$. There are two trivial ideals $C(x)=0$ and $C(x)=\overline{R}$, the corresponding $\phi$-cyclic subspace are $C=0$ and $C=\mathbb{R}^n$. To find non-trivial $\phi$-cyclic subspaces, we make use of the homomorphism theorems, which is a standard technique in algebra. Let $\pi$ be the natural homomorphism from $\mathbb{R}[x]$ to $\overline{R}$, ker$\pi=\phi(x)\mathbb{R}[x]$. We write $\phi(x)\mathbb{R}[x]$ by $<\phi(x)>$. Let $N$ be an ideal of $\mathbb{R}[x]$ satisfying
\begin{equation}
<\phi(x)>\subset N \subset \mathbb{R}[x] \xrightarrow{\quad\pi\quad} \overline{R}=\mathbb{R}[x] / <\phi(x)>.
\end{equation}
Since $\mathbb{R}[x]$ is a principal ideal domain, then $N=<g(x)>$ is a principal ideal generated by a monic polynomial $g(x)\in \mathbb{R}[x]$. It is easy to see that
\begin{equation*}
<\phi(x)>\subset <g(x)>\Leftrightarrow g(x)|\phi(x)\ \text{in}\ \mathbb{R}[x].
\end{equation*}
It follows that all ideals $N$ satisfying (3.2) are given by
\begin{equation*}
\{<g(x)>\Big|\ g(x)\in \mathbb{R}[x]\ \text{is monic and}\ g(x)|\phi(x)\}.
\end{equation*}
We write by $<g(x)>$ mod $\phi(x)$, the image of $<g(x)>$ under $\pi$, i.e.
\begin{equation*}
<g(x)>\ \text{mod}\ \phi(x)=\pi(<g(x)>)
\end{equation*}
It is easy to check
\begin{equation}
<g(x)>\ \text{mod}\ \phi(x)=\{a(x)g(x)\ |\ a(x)\in \mathbb{R}[x] \text{\ and deg} a(x)+\text{deg}g(x)<n\},
\end{equation}
more precisely, which is a representative elements set of $<g(x)>$ mod $\phi(x)$. By homomorphism theorem in ring theory, all ideals of $\overline{R}$ given by
\begin{equation}
\{<g(x)> \text{\ mod\ } \phi(x) \ \Big|\ g(x) \in \mathbb{R}[x] \text{\ is monic and\ } g(x)|\phi(x)\}.
\end{equation}
Let $d$ be the number of monic divisors of $\phi(x)$ in $\mathbb{R}[x]$, we have\\

\begin{cor}\label{co3.1} The number of $\phi$-cyclic subspace of $\mathbb{R}^n$ is $d$.
\end{cor}

 Next, we discuss $\phi$-cyclic lattice, which is the geometric analogy of  cyclic code.   The $\phi$-cyclic  subspace of $ \mathbb{R}^{n}$ maybe regarded as the algebraic analogy of cyclic code. Let the quotient rings $R$ and $\overline{R}$ given by (2.8).  A $R$-module is an Abel group $\wedge$ such that there is an operator $\lambda \alpha\in \wedge$ for all $\lambda\in R$ and $\alpha\in \wedge$, satisfying $1\cdot \alpha=\alpha$ and $(\lambda_1 \lambda_2)\alpha=\lambda_1 (\lambda_2 \alpha)$. It is easy to see that $\overline{R}$ is a $R$-module, if $\wedge\subset\overline{R}$ and $\wedge$ is a $R$-module, then $\wedge$ is called a $R$-submodule of $\overline{R}$. All $R$-modules we discuss here are $R$-submodule of $\overline{R}$. On the other hand, if $I\subset R$, then $I$ is an ideal of $R$, if and only if $I$ is a $R$-module. Let $\alpha\in\overline{R}$, the cyclic $R$-module generated by $\alpha$ be defined by
\begin{equation}
R\alpha=\{\lambda \alpha\ |\ \lambda\in R\}.
\end{equation}
If there are finitely many polynomials $\alpha_1,\alpha_2,\dots,\alpha_k$ in $\overline{R}$ such that $\wedge=R\alpha_1+R\alpha_2+\cdots+R\alpha_k$, then $\wedge$ is called a finitely generated $R$-module, which is a $R$-submodule of $\overline{R}$.

Now, if $L\subset \mathbb{R}^n$ is a $\phi$-cyclic lattice, $g\in \mathbb{R}^n$, $H^*(g)$ is the ideal matrix generated by vector $g$, and $L(H^*(g))$ is the lattice generated by $H^*(g)$. It is easy to show that any $L(H^*(g))$ is a $\phi$-cyclic lattice and
\begin{equation}
L(H^*(g))\subset L,\ \text{whenever}\ g\in L,
\end{equation}
which implies that $L(H^*(g))$ is the smallest $\phi$-cyclic lattice of which contains vector $g$. Therefore, we call $L(H^*(g))$ is a minimal $\phi$-cyclic lattice in $\mathbb{R}^n$.\\




\begin{lem}\label{lm3.3} There is a one to one correspondence between the minimal $\phi$-cyclic lattice in $\mathbb{R}^n$ and the cyclic $R$-submodule in $\overline{R}$, namely,
\begin{equation*}
t(Rg(x))=L(H^*(g)),\ \text{for all}\ g(x)\in \overline{R}
\end{equation*}
and
\begin{equation*}
t^{-1}(L(H^*(g)))=Rg(x),\ \text{for all}\ g\in \mathbb{R}^n.
\end{equation*}
\end{lem}
\begin{proof} Let $b(x)\in R$, by lemma \ref{lm2.1}, we have
\begin{equation*}
t(b(x)g(x))=H^*(b)g=H^*(g)b\in L(H^*(g)),
\end{equation*}
and $t(Rg(x))\subset L(H^*(g))$. Conversely, if $\alpha\in L(H^*(g))$, and $\alpha=H^*(g)b$ for some integer vector $b$, by lemma \ref{lm2.1} again, we have $b(x)g(x)\in Rg(x)$, and $t(b(x)g(x))=\alpha$. This implies that $L(H^*(g))\subset t(Rg(x))$, and
\begin{equation*}
t(Rg(x))=L(H^*(g)).
\end{equation*}
The lemma follows immediately.

\end{proof}

Suppose $L=L(\beta_1,\beta_2,\dots,\beta_m)$ is arbitrary $\phi$-cyclic lattice, where $B=[\beta_1,\beta_2,\dots,\beta_m]_{n\times m}$ is the generated matrix of $L$.  $L$ may be expressed as the sum of finitely many minimal $\phi$-cyclic lattices, in fact, we have
\begin{equation}
L=L(H^*(\beta_1))+L(H^*(\beta_2))+\cdots+L(H^*(\beta_m)).
\end{equation}

To state and prove our main results, first, we give a definition of prime spot in $\mathbb{R}^n$.\\

\begin{defn}\label{def3.2} Let $g\in \mathbb{R}^n$, and $g(x)=t^{-1}(g)\in \overline{R}$. If $(g(x),\phi(x))=1$ in $\mathbb{R}[x]$, we call $g$ is a prime spot of $\mathbb{R}^n$.

\end{defn}
By (v) of Theorem \ref{th2}, $g\in \mathbb{R}^n$ is a prime spot if and only if $H^*(g)$ is an invertible matrix, thus the minimal $\phi$-cyclic lattice $L(H^*(g))$ generated by a prime spot is a full-rank lattice.\\

\begin{lem}\label{lm3.4} Let $g$ and $f$ be two prime spots of $\mathbb{R}^n$, then $L(H^*(g))+L(H^*(f))$ is a full-rank $\phi$-cyclic lattice.

\end{lem}
\begin{proof} According to lemma \ref{lm1.4}, it is sufficient to show that
\begin{equation}
\text{rank}\big(L(H^*(g))\cap L(H^*(f))\big)=\text{rank}\big(L(H^*(g))\big)=n.
\end{equation}
In fact, we should prove in general
\begin{equation}
L(H^*(g)\cdot H^*(f))\subset L(H^*(g))\cap L(H^*(f)).
\end{equation}
Since $H^*(g)\cdot H^*(f)$ is an invertible matrix, then rank$\big(L(H^*(g)\cdot H^*(f))\big)=n$, and (3.8) follows immediately.

To prove (3.9), we note that
\begin{equation*}
L(H^*(g)\cdot H^*(f))=L(H^*(g*f)).
\end{equation*}
It follows that
\begin{equation*}
t^{-1}\big(L(H^*(g)\cdot H^*(f))\big)=Rg(x)f(x).
\end{equation*}
It is easy to see that
\begin{equation*}
Rg(x)f(x)\subset Rg(x)\cap Rf(x).
\end{equation*}
Therefore, we have
\begin{equation*}
L(H^*(g)\cdot H^*(f))=t(Rg(x)f(x))\subset L(H^*(g))\cap L(H^*(f)).
\end{equation*}
This is the proof of lemma \ref{lm3.4}.

\end{proof}

It is worth to note that (3.9) is true for more general case, do not need the condition of prime spot.\\

\begin{cor}\label{co3.2} Let $\beta_1,\beta_2,\dots,\beta_m$ be arbitrary $m$ vectors in $\mathbb{R}^n$, then we have
\begin{equation}
L(H^*(\beta_1) H^*(\beta_2)\cdots H^*(\beta_m))\subset L(H^*(\beta_1))\cap L(H^*(\beta_2))\cap\cdots\cap L(H^*(\beta_m)).
\end{equation}
\end{cor}
\begin{proof} If $\beta_1,\beta_2,\dots,\beta_m$ are integer vectors, then (3.10) is trivial. For the general case, we write
\begin{equation*}
L(H^*(\beta_1)\cdot H^*(\beta_2)\cdots H^*(\beta_m))=L(H^*(\beta_1 *\beta_2 *\cdots *\beta_m)),
\end{equation*}
where $\beta_1 *\beta_2 *\cdots *\beta_m$ is the $\phi$-convolutional product, then
\begin{equation*}
t^{-1}\big(L(H^*(\beta_1)\cdots H^*(\beta_m))\big)=R \beta_1(x)\beta_2(x)\cdots \beta_m(x).
\end{equation*}
Since
\begin{equation*}
R \beta_1(x)\beta_2(x)\cdots \beta_m(x)\subset R\beta_1(x)\cap R\beta_2(x)\cap\cdots\cap R\beta_m(x).
\end{equation*}
It follows that
\begin{equation*}
L(H^*(\beta_1) H^*(\beta_2)\cdots H^*(\beta_m))\subset L(H^*(\beta_1))\cap L(H^*(\beta_2))\cap\cdots\cap L(H^*(\beta_m)).
\end{equation*}
We have this corollary.

\end{proof}

By lemma \ref{lm3.4}, we also have the following assertion.

\begin{cor}\label{co3.3} Let $\beta_1,\beta_2,\dots,\beta_m$ be $m$ prime spots of $\mathbb{R}^n$, then $L(H^*(\beta_1))+L(H^*(\beta_2))+\cdots+L(H^*(\beta_m))$ is a full-rank $\phi$-cyclic lattice.
\end{cor}
\begin{proof} It follows immediately from corollary \ref{co1.3}.
\end{proof}

Our main result in this paper is to establish the following one to one correspondence between $\phi$-cyclic lattices in $\mathbb{R}^n$ and finitely generated $R$-modules in $\overline{R}$.

\begin{thm}\label{th4} Let $\wedge=R\alpha_1(x)+R\alpha_2(x)+\cdots+R\alpha_m(x)$ be a finitely generated $R$-module in $\overline{R}$, then $t(\wedge)$ is a $\phi$-cyclic lattice in $\mathbb{R}^n$. Conversely, if $L\subset \mathbb{R}^n$ is a $\phi$-cyclic lattice in $\mathbb{R}^n$, then $t^{-1}(L)$ is a finitely generated $R$-module in $\overline{R}$, that is a one to one correspondence.
\end{thm}
\begin{proof}
 If $\wedge$ is a finitely generated $R$-module, by lemma \ref{lm3.3}, we have
\begin{equation*}
t(\wedge)=t(R\alpha_1(x)+\cdots+R\alpha_m(x))=L(H^*(\alpha_1))+L(H^*(\alpha_2))+\cdots+L(H^*(\alpha_m)).
\end{equation*}
The main difficult is to show that $t(\wedge)$ is a lattice of $\mathbb{R}^n$, we require a surgery to embed $t(\wedge)$ into a full-rank lattice. To do this, let $(\alpha_i(x),\phi(x))=d_i(x)$, $d_i(x)\in \mathbb{Z}[x]$, and $\beta_i(x)=\alpha_i(x)/d_i(x)$, $1\leqslant i\leqslant m$. Since $\phi(x)$ has no multiple roots by assumption, then $(\beta_i(x),\phi(x))=1$ in $\mathbb{R}[x]$. In other words, each $t(\beta_i(x))=\beta_i$ is a prime spot. It is easy to verify $R\alpha_i(x)\subset R\beta_i(x)\ (1\leqslant i\leqslant m)$, thus we have
\begin{equation*}
t(\wedge)\subset L(H^*(\beta_1))+L(H^*(\beta_2))+\cdots+L(H^*(\beta_m)).
\end{equation*}
By corollary \ref{co3.3} and corollary \ref{co1.1}, we have $t(\wedge)$ is $\phi$-cyclic lattice. Conversely, if $L\subset \mathbb{R}^n$ is a $\phi$-cyclic lattice of $\mathbb{R}^n$, and $L=L(\beta_1,\beta_2,\dots,\beta_m)$, by (3.7), we have
\begin{equation*}
t^{-1}(L)=R\beta_1(x)+R\beta_2(x)+\cdots+R\beta_m(x),
\end{equation*}
which is a finitely generated $R$-module in $\overline{R}$. We complete the proof of Theorem \ref{th4}.

\end{proof}

As we introduced in abstract, since $R$ is a Noether ring, then $I\subset R$ is an ideal if and only if $I$ is a finitely generated $R$-module. On the other hand, if $I\subset R$ is an ideal, then $t(I)\subset \mathbb{Z}^n$ is a discrete subgroup of $\mathbb{Z}^n$, thus $t(I)$ is a lattice, we define\\

\begin{defn}\label{def3.3}
 Let $I\subset R$ be an ideal, $t(I)$ is called the $\phi$-ideal lattice.
\end{defn}
Ideal lattice was first appeared in \cite{L1} (see Definition 3.1 of \cite{L1}). As a direct consequences of Theorem \ref{th4}, we have

\begin{cor}\label{co3.4}
Let $L\subset \mathbb{R}^n$ be a subset, then $L$ is a $\phi$-cyclic lattice if and only if
\begin{equation*}
L=L(H^*(\beta_1))+L(H^*(\beta_2))+\cdots+L(H^*(\beta_m)),
\end{equation*}
where $\beta_i\in \mathbb{R}^n$ and $m\leqslant n$. Furthermore, $L$ is a $\phi$-ideal lattice if and only if every $\beta_i\in \mathbb{Z}^n$, $1\leqslant i\leqslant m$.
\end{cor}
\begin{cor}\label{co3.5} Suppose that $\phi(x)$ is an irreducible polynomial in $\mathbb{Z}[x]$, then any non-zero ideal $I$ of $R$ defines a full-rank $\phi$-ideal lattice $t(I)\subset \mathbb{Z}^n$.
\end{cor}

\begin{proof} Let $I\subset R$ be a non-zero ideal, then we have $I=R\alpha_1(x)+R\alpha_2(x)+\cdots+R\alpha_m(x)$, where $\alpha_i(x)\in R$ and $(\alpha_i(x),\phi(x))=1$. It follows that
\begin{equation*}
t(I)=L(H^*(\alpha_1))+L(H^*(\alpha_2))+\cdots+L(H^*(\alpha_m)).
\end{equation*}
Since each $\alpha_i$ is a prime spot, we have rank$(t(I))=n$ by corollary \ref{co3.3}, and the corollary follows at once.

\end{proof}

According to Definition 3.1 of \cite{L1}, we have proved that any an ideal of $R$ corresponding to a $\phi$-ideal lattice, which just is a $\phi$-cyclic integer lattice under the more general rotation matrix $H=H_{\phi}$. Cyclic lattice and ideal lattice were introduced in \cite{L1} and \cite{D2} respectively to improve the space complexity of lattice based cryptosystems. Ideal lattices allow to represent a lattice using only two polynomials. Using such lattices, class lattice based cryptosystems can diminish their space complexity from $O(n^2)$ to $O(n)$. Ideal lattices also allow to accelerate computations using the polynomial structure. The original structure of Micciancio's matrices uses the ordinary circulant matrices and allows for an interpretation in terms of arithmetic in polynomial ring $\mathbb{Z}[x]/<x^n-1>$. Lyubashevsky and Micciancio \cite{L1} latter suggested to change the ring to $\mathbb{Z}[x]/<\phi(x)>$ with an irreducible $\phi(x)$ over $\mathbb{Z}[x]$. Our results here suggest to change the ring to $\mathbb{Z}[x]/<\phi(x)>$ with any a polynomial $\phi(x)$. There are many works are subsequent to Micciancio \cite{D2} and Lyubashevsky and Micciancio \cite{L1}, such as [7], \cite{D4}, \cite{C}, \cite{T}, \cite{p2}, [17] and \cite{d2}.

\begin{example}\label{ex3.1}
It is interesting to find some examples of $\phi$-cyclic lattices in an algebraic number field $K$. Let $Q$ be rational number field, without loss of generality, an algebraic number field $K$ of degree $n$ is just $K=Q(w)$, where $w=w_i$ is a root of $\phi(x)$. If all $Q(w_i)\subset \mathbb{R}\ (1\leqslant i\leqslant n)$, then $K$ is called a totally real algebraic number field. Let $O_K$ be the ring of algebraic integers of $K$, and $I\subset O_K$ be an ideal, $I\neq 0$. Since there is an integral basis $\{\alpha_1,\alpha_2,\dots,\alpha_n\}\subset I$ such that
\begin{equation*}
I=\mathbb{Z}\alpha_1+\mathbb{Z}\alpha_2+\cdots+\mathbb{Z}\alpha_n.
\end{equation*}
We may regard every ideal of $O_K$ as a lattice in $Q^n$, our assertion is that every nonzero ideal of $O_K$ is corresponding to a full-rank $\phi$-cyclic lattice of $Q^n$. To see this example, let
\begin{equation*}
Q[w]=\{\sum_{i=0}^{n-1} a_i w^i\ |\ a_i\in Q\}.
\end{equation*}
It is known that $K=Q[w]$, thus every $\alpha\in K$ corresponds to a vector $\overline{\alpha}\in Q^n$ by
\begin{equation*}
\alpha=\sum_{i=0}^{n-1} a_i w^i \xrightarrow{\quad \tau\quad} \overline{\alpha}=\begin{pmatrix}
     a_0 \\
     a_1 \\
     \vdots \\
     a_{n-1}
\end{pmatrix}\in \mathbb{Q}^n.
\end{equation*}
If $I\subset O_K$ is an ideal of $O_K$ and $I=\mathbb{Z}\alpha_1+\mathbb{Z}\alpha_2+\cdots+\mathbb{Z}\alpha_n$, let $B=[\overline{\alpha_1},\overline{\alpha_2},\dots,\overline{\alpha_n}]\in Q^{n\times n}$, which is full-rank matrix. We have $\tau(I)=L(B)$ is a full-rank lattice. It remains to show that $\tau(I)$ is a $\phi$-cyclic lattice, we only prove that if $\alpha \in I\Rightarrow H\overline{\alpha}\in \tau(I)$. Suppose that $\alpha\in I$, then $w\alpha \in I$. It is easy to verify that $\tau(w)=e_2$ (see (2.7)) and
\begin{equation*}
\tau(w\alpha)=\tau(w)*\tau(\alpha)=H\overline{\alpha}\in \tau(I).
\end{equation*}
This means that $\tau(I)$ is a $\phi$-cyclic lattice of $Q^n$, which is a full-rank lattice.
\end{example}

\section{Smoothing Parameter}

As application of the algebraic structure of $\phi$-cyclic lattice, we show that an explicit upper bound of the smoothing parameter for the $\phi$-cyclic lattices. Firstly, we introduce some basic notations.

A Gauss function $\rho_{s,c}(x)$ in $\mathbb{R}^n$ is given by
\begin{equation}
\rho_{s,c}(x)=e^{-\pi |x-c|^2/s^2},
\end{equation}
where $x\in \mathbb{R}^n$, $c\in \mathbb{R}^n$ and $s>0$ is a positive real number. $\rho_{s,c}(x)$ is called the Gauss function around original point $c$ with parameter $s$. It is easy to see that
\begin{equation*}
\int_{\mathbb{R}^n} \rho_{s,c}(x)\mathrm{d} x=s^n.
\end{equation*}
Thus we may define a probability density function $D_{s,c}(x)$ by
\begin{equation}
D_{s,c}(x)=\rho_{s,c}(x)/ \int_{\mathbb{R}^n} \rho_{s,c}(x)\mathrm{d} x=\rho_{s,c}(x)/s^n.
\end{equation}
Suppose $L\subset \mathbb{R}^n$ is a lattice, let
\begin{equation}
D_{s,c}(L)=\sum\limits_{x\in L} D_{s,c}(x),\ \rho_{s,c}(L)=\sum\limits_{x\in L} \rho_{s,c}(x).
\end{equation}
The discrete Gauss distribution over $L$ is a probability distribution $D_{L,s,c}$ over $L$ given by
\begin{equation}
D_{L,s,c}(x)=\frac{D_{s,c}(x)}{D_{s,c}(L)}=\frac{\rho_{s,c}(x)}{\rho_{s,c}(L)}.
\end{equation}
If $c=0$ is the zero vector of $\mathbb{R}^n$, we write $\rho_{s,0}(x)=\rho_{s}(x)$, $\rho_{s,0}(L)=\rho_{s}(L)$, $D_{s,0}(x)=D_{s}(x)$ and $D_{s,0}(L)=D_{s}(L)$. Suppose that $L$ is a full-rank lattice and $L^*$ is its dual lattice, we define the smoothing parameter $\eta_{\epsilon}(L)$ of $L$ to be the smallest $s$ such that $\rho_{1/s}(L^*)\leqslant 1+\epsilon$, more precisely,
\begin{equation}
\eta_{\epsilon}(L)=\min\{s:\ s>0\ \text{and}\ \rho_{1/s}(L^*)\leqslant 1+\epsilon\},
\end{equation}
where $\epsilon>0$ is a positive number. Notice that $\rho_{1/s}(L^*)$ is a continuous and strictly decreasing function of $s$, thus the smoothing parameter $\eta_{\epsilon}(L)$ is a continuous and strictly decreasing function of $\epsilon$.

Let $L=L(\beta_1,\beta_2,\dots,\beta_n)\subset \mathbb{R}^n$ be a full-rank lattice with a basis $\beta_1,\beta_2,\dots,\beta_n$, the fundamental region $P(L)$ is given by
\begin{equation}
P(L)=\{\sum\limits_{i=1}^{n} a_i \beta_i | 0\leqslant a_i< 1,\ 1\leqslant i\leqslant n\}.
\end{equation}
Suppose that $X$ and $Y$ are two discrete random variables on $\mathbb{R}^n$, the statistical distance between $X$ and $Y$ over $L$ is defined by
\begin{equation}
\triangle(X,Y)=\frac{1}{2} \sum\limits_{a\in L}|P\{X=a\}-P\{Y=a\}|.
\end{equation}
If $X$ and $Y$ are continuous random variables with probability density function $T_1$ and $T_2$ respectively, then $\triangle (X,Y)$ is defined by
\begin{equation}
\triangle (X,Y)=\frac{1}{2}\int_{\mathbb{R}^n} |T_1(z)-T_2(z)| \mathrm{d} z.
\end{equation}

The smoothing parameter was introduced by Micciancio and Regev in \cite{D3}, which plays an important role in the statistical information of lattices. An important property of smoothing parameter is for any lattice $L=L(B)$ and any $\epsilon>0$, the statistical distance between $D_s$ mod $L$ and the uniform distribution over the fundamental region $P(L)$ is at most $\frac{1}{2}(\rho_{1/s}(L(B)^*))$. More precisely, for any $\epsilon>0$ and any $s\geqslant \eta_{\epsilon} (L(B))$, the statistical distance is at most $\frac{1}{2}\epsilon$, namely
\begin{equation}
\triangle \big(D_{s,c}\ \text{mod}\ L,\ U(P(L)) \big)\leqslant \frac{\varepsilon}{2}.
\end{equation}

\begin{lem}\label{lm4.1}Let $L\subset \mathbb{R}^n$ be a full-rank lattice, we have
\begin{equation}
\eta_{2^{-n}} (L)\leqslant \sqrt{n}/\lambda_1 (L^*),
\end{equation}
where $L^*$ is the dual lattice of $L$, and $\lambda_1 (L^*)$ is the minimum distance of $L^*$.
\end{lem}
\begin{proof} See lemma 3.2 of \cite{D3}, or [3].
\end{proof}

\begin{lem}\label{lm4.2} Suppose that $L_1$ and $L_2$ are two full-rank lattices in $\mathbb{R}^n$, and $L_1 \subset L_2$, then for any $\epsilon>0$, we have
\begin{equation}
\eta_{\epsilon}(L_2)\leqslant \eta_{\epsilon}(L_1).
\end{equation}
\end{lem}
\begin{proof}Let $\eta_{\epsilon}(L_1)=s$, we are to show that $\eta_{\epsilon}(L_2)\leqslant s$. Since
\begin{equation*}
\rho_{1/s}(L_1^*)=1+\epsilon,\ \text{and}\ \sum\limits_{x\in L_1^*}e^{-\pi s^2 |x|^2}=1+\epsilon.
\end{equation*}
It is easy to check that $L_2^*\subset L_1^*$, it follows that
\begin{equation*}
1+\epsilon=\sum\limits_{x\in L_1^*}e^{-\pi s^2 |x|^2}\geqslant \sum\limits_{x\in L_2^*}e^{-\pi s^2 |x|^2},
\end{equation*}
which implies
\begin{equation*}
\rho_{1/s}(L_2^*)\leqslant 1+\epsilon,
\end{equation*}
and $\eta_{\epsilon}(L_2)\leqslant s=\eta_{\epsilon}(L_1)$, thus we have lemma \ref{lm4.2}.

\end{proof}

According to (2.4), the ideal matrix $H^*(f)$ with input vector $f\in \mathbb{R}^n$ is just the ordinary circulant matrix when $\phi(x)=x^n-1$. Next lemma shows that the transpose of a circulant matrix is still a circulant matrix. For any $g=\begin{pmatrix}
     g_0 \\
     g_1 \\
     \vdots \\
     g_{n-1}
\end{pmatrix}\in \mathbb{R}^n$, we denote $\overline{g}=\begin{pmatrix}
     g_{n-1} \\
     g_{n-2} \\
     \vdots \\
     g_{0}
\end{pmatrix}$, which is called the conjugation of $g$.\\

\begin{lem}\label{lm4.3} Let $\phi(x)=x^n-1$, then for any $g=\begin{pmatrix}
     g_0 \\
     g_1 \\
     \vdots \\
     g_{n-1}
\end{pmatrix}\in \mathbb{R}^n$, we have
\begin{equation}
(H^*(g))'=H^*(H\overline{g}).
\end{equation}
\end{lem}
\begin{proof}  Since $\phi(x)=x^n-1$, then $H=H_{\phi}$ (see(2.3)) is an orthogonal matrix, and we have $H^{-1}=H^{n-1}=H'$. We write $H_1=H'=H^{-1}$. The following identity is easy to verify
\begin{equation*}
H^*(g)=\begin{pmatrix}
     \overline{g}'H_1 \\
     \overline{g}'H_1^2 \\
     \vdots \\
     \overline{g}'H_1^n
\end{pmatrix}
\end{equation*}
It follows that
\begin{equation*}
(H^*(g))'=[H\overline{g},H(H\overline{g}),\dots,H^{n-1}(H\overline{g})]=H^*(H\overline{g}),
\end{equation*}
and we have the lemma.

\end{proof}

\begin{lem}\label{lm4.4}
 Suppose that $g\in \mathbb{R}^n$ and the circulant matrix $H^*(g)$ is invertible. Let $A=(H^*(g))'H^*(g)$, then all characteristic values of $A$ are given by
\begin{equation*}
\{|g(\theta_1)|^2,|g(\theta_2)|^2,\dots,|g(\theta_n)|^2\},
\end{equation*}
where $\theta_i^n=1\ (1\leqslant i\leqslant n)$ are the $n$-th roots of unity.
\end{lem}

\begin{proof}  By lemma \ref{lm4.3} and (ii) of Theorem \ref{th2}, we have
\begin{equation*}
A=H^*(H\overline{g})H^*{g}=H^*(H^*(H\overline{g})g)=H^*(g''),
\end{equation*}
where $g''=H^*(H\overline{g})g$. Let $g''(x)=t^{-1}(g'')$ is the corresponding polynomial of $g''$. By (iii) of Theorem \ref{th2} all characteristic values of $A$ are given by
\begin{equation}
\{g''(\theta_1),g''(\theta_2),\dots,g''(\theta_n)\},\ \theta_i^n=1,\ 1\leqslant i\leqslant n.
\end{equation}
Let $g=\begin{pmatrix}
     g_0 \\
     g_1 \\
     \vdots \\
     g_{n-1}
\end{pmatrix}\in \mathbb{R}^n$. It is easy to see that
\begin{equation*}
g''(x)=\sum\limits_{i=0}^{n-1}g_i^2+(\sum\limits_{i=0}^{n-1}g_i g_{1-i})x+\cdots+(\sum\limits_{i=0}^{n-1}g_i g_{(n-1)-i})x^{n-1}=|g(x)|^2,
\end{equation*}
where $g_{-i}=g_{n-i}$ for all $1\leqslant i\leqslant n-1$, then the lemma follows at once.

\end{proof}

By definition \ref{def3.2}, if $g\in \mathbb{R}^n$ is a prime spot, then there is a unique polynomial $u(x)\in \overline{R}$ such that $u(x)g(x)\equiv 1$ (mod $\phi(x)$). We define a new vector $T_g$ and its corresponding polynomial $T_g(x)$ by
\begin{equation}
T_g=H\overline{u},\ \text{and}\ T_g(x)=t^{-1}(H\overline{u}).
\end{equation}
If $g\in \mathbb{Z}^n$ is an integer vector, then $T_g\in \mathbb{Z}^n$ is also an integer vector, and $T_g(x)\in \mathbb{Z}[x]$ is a polynomial with integer coefficients. Our main result on smoothing parameter is the following theorem.\\

\begin{thm}\label{th5} Let $\phi(x)=x^n-1$, $L\subset \mathbb{R}^n$ be a full-rank $\phi$-cyclic lattice, then for any prime spots $g\in L$, we have
\begin{equation}
\eta_{2^{-n}}(L)\leqslant \sqrt{n} (\min\{|T_g(\theta_1)|,|T_g(\theta_2)|,\dots,|T_g(\theta_n)|\})^{-1},
\end{equation}
where $\theta_i^n=1$, $1\leqslant i\leqslant n$, and $T_g(x)$ is given by (4.14).
\end{thm}

\begin{proof}  Let $g\in L$ be a prime spot, by  lemma \ref{lm4.2}, we have
\begin{equation}
L(H^*(g))\subset L\Rightarrow \eta_{\epsilon}(L)\leqslant \eta_{\epsilon}(L(H^*(g))),\ \forall \epsilon>0.
\end{equation}
To estimate the smoothing parameter of $L(H^*(g))$, the dual lattice of $L(H^*(g))$ is given by
\begin{equation*}
L(H^*(g))^*=L((H^*(u))')=L(H^*(H\overline{u}))=L(H^*(T_g)),
\end{equation*}
where $u(x)\in\overline{R}$ and $u(x)g(x)\equiv 1$ (mod $x^n-1$), and $T_g$ is given by (4.14). Let $A=(H^*(T_g))'H^*(T_g)$, by lemma \ref{lm4.4}, all characteristic values of $A$ are
\begin{equation*}
\{|T_g(\theta_1)|^2,|T_g(\theta_2)|^2,\dots,|T_g(\theta_n)|^2\}.
\end{equation*}
By lemma \ref{lm1.2}, the minimum distance $\lambda_1 (L(H^*(g))^*)$ is bounded by
\begin{equation}
\lambda_1 (L(H^*(g))^*)\geqslant \min \{|T_g(\theta_1)|,|T_g(\theta_2)|,\dots,|T_g(\theta_n)|\}.
\end{equation}
Now, Theorem \ref{th5} follows from lemma \ref{lm4.1} immediately.
\end{proof}

Let $L=L(B)$ be a full-rank lattice and $B=[\beta_1,\beta_2,\dots,\beta_n]$. We denote by $B^*=[\beta_1^*,\beta_2^*,\dots,\beta_n^*]$ the Gram-Schmidt orthogonal vectors $\{\beta_i^*\}$ of the ordered basis $B=\{\beta_i\}$. It is a well-known conclusion that
\begin{equation*}
\lambda_1(L)\geqslant |B^*|=\min\limits_{1\leqslant i\leqslant n} |\beta_i^*|,
\end{equation*}
which yields by lemma \ref{lm4.1} the following upper bound
\begin{equation}
\eta_{2^{-n}}(L)\leqslant \sqrt{n} |B_0^*|^{-1},
\end{equation}
where $B_0^*$ is the orthogonal basis of dual lattice $L^*$ of $L$.

For a $\phi$-cyclic lattice $L$, we observe that the upper bound (4.17) is always better than (4.18) by numerical testing, we give two examples here.

\begin{example}
Let $n=3$ and $\phi(x)=x^3-1$, the rotation matrix $H$ is
\begin{equation*}
H=\begin{pmatrix}
     0 & 0 & 1 \\
     1 & 0 & 0 \\
     0 & 1 & 0
\end{pmatrix}.
\end{equation*}
We select a $\phi$-cyclic lattice $L=L(B)$, where
\begin{equation*}
B=\begin{pmatrix}
     1 & 1 & 1 \\
     0 & 1 & 1 \\
     0 & 0 & 1
\end{pmatrix}.
\end{equation*}
Since $L=\mathbb{Z}^3$, thus $L$ is a $\phi$-cyclic lattice. It is easy to check
\begin{equation*}
|B_0^*|=\min\limits_{1\leqslant i\leqslant 3}|\beta_i^*|=\frac{\sqrt{3}}{3}.
\end{equation*}
On the other hand, we randomly find a prime spot $g=\begin{pmatrix}
0\\
0\\
1
\end{pmatrix}\in L$ and $g(x)=x^2$. Since $xg(x)\equiv 1$ (mod $x^3-1$), we have $T_g(x)=x^2$, it follows that $|T_g(\theta_1)|=|T_g(\theta_2)|=|T_g(\theta_3)|=1$, and
\begin{equation*}
\min\limits_{1\leqslant i\leqslant 3}|T_g(\theta_i)|^{-1}\leqslant |B_0^*|^{-1}=\sqrt{3}.
\end{equation*}
\end{example}

\begin{example}
Let $n=4$ and $\phi(x)=x^4-1$, the rotation matrix $H$ is
\begin{equation*}
H=\begin{pmatrix}
     0 & 0 & 0 & 1 \\
     1 & 0 & 0 & 0 \\
     0 & 1 & 0 & 0 \\
     0 & 0 & 1 & 0
\end{pmatrix}.
\end{equation*}
We select a $\phi$-cyclic lattice $L=L(B)$, where
\begin{equation*}
B=\begin{pmatrix}
     1 & 1 & 1 & 1 \\
     0 & 1 & 1 & 1 \\
     0 & 0 & 1 & 1 \\
     0 & 0 & 0 & 1
\end{pmatrix}.
\end{equation*}
Since $L=\mathbb{Z}^4$, thus $L$ is a $\phi$-cyclic lattice. It is easy to check
\begin{equation*}
|B_0^*|=\min\limits_{1\leqslant i\leqslant 4}|\beta_i^*|=\frac{1}{2}.
\end{equation*}
On the other hand, we randomly find a prime spot $g=\begin{pmatrix}
-2 \\
1 \\
0 \\
0
\end{pmatrix}\in L$ and $g(x)=x-2$. Since $(\frac{1}{7}x^3-\frac{1}{7}x^2-\frac{2}{7}x-\frac{5}{7})g(x)\equiv 1$ (mod $x^4-1$), we have $T_g(x)=-\frac{2}{7}x^3-\frac{1}{7}x^2+\frac{1}{7}x-\frac{5}{7}$, it follows that $|T_g(\theta_1)|=1$, $|T_g(\theta_2)|=|T_g(\theta_3)|=|T_g(\theta_4)|=\frac{5}{7}$, and
\begin{equation*}
\min\limits_{1\leqslant i\leqslant 4}|T_g(\theta_i)|^{-1}=\frac{7}{5}\leqslant |B_0^*|^{-1}=2.
\end{equation*}
\end{example}

 \vspace{-0.4cm}{\footnotesize}


\begin{thebibliography}{aa}

\bibitem{M1}   M. Ajtai, Generating Hard Instances of the Short Basis Problem. In Proc. of 28th STOC,  99-108, 1996.

\bibitem{M2}  M. Ajtai, C. Dwork, A public-key cryptosystem with worst-case/average-case equivalence.  In Proc. of 29th STOC,  284-293, 1997.

\bibitem{W} W. Banaszczyk, New bounds in some transference theorems in the geometry of numbers. Mathematische Annalen, 296(4): 625-635,1993.

\bibitem{J1} J. W. S. Cassels, An introduction to the geometry of numbers. Berlin, Heidelberg, New York: Springer, 1971.

\bibitem{J2} J. W. S. Cassels, Introduction to Diophantine Approximation, Cambridge University Press, 1963.

\bibitem{p}  P. J. Davis, Circulant matrices. 2nd Edition, Chelsea Publishing, New York, 1994.

\bibitem{U} U. Feige, D. Micciancio, The inapproximability of lattice and coding problems with preprocessing. Journal of Computer and System Sciences , 69(1):45-67, 2004.

\bibitem{G}   C. Gentry, Fully homomorphic encryption using ideal lattices. Stoc, 2009.

\bibitem{L} J. H. V. Lint,  Introduction to Coding Theory. Springer-Verlag, 1999.

\bibitem{L1} V. Lyubashevsky, D. Micciancio, Generalized Compact Knapsacks Are Collision Resistant. Proceedings of the 33rd international conference on Automata,  Languages and Programming - Proceedings of ICALP 2006, Vencie, Italy. Springer LNCS 4052,  144-155.

\bibitem{D1}  D. Micciancio,  The hardness of the closest vector problem with preprocessing. IEEE  Transactions on  Information Theory, 47(3): 1212-1215,  2001.

\bibitem{D2}D. Micciancio, Generalized compact knapsacks, cyclic lattices, and efficient one-way functions from worst-case complexity assumptions: (extended abstract). Annual Symposium on Foundations of Computer Science, 2002.

\bibitem{D3} D. Micciancio,  O. Regev,  Worst-Case to Average-Case Reductions Based on Gaussian Measures.  SIAM J. Comput, 37(1): 267-302, 2007.

\bibitem{D4} D. Micciancio, O. Regev, Lattice-based Cryptography. In D.J.Bernstein, J. Buchmann, E. Dahmen(eds) Post-Quantum Cryptography, Springer Berlin Heidelberg, 147-191, 2009.

\bibitem{C} C. Peikert, A Decade of Lattice Cryptography. Foundations and trends in theoretical computer science,  2016.

\bibitem{T}  T. Plantard, M. Schneider, Creating a challenge for ideal lattices, 1-17, 2013.

\bibitem{p2} P. K. Pradhan, S.Rakshit, S.Datta, Lattice based cryptography: Its applications, areas of interest and future scope. Proceedings of the Third International Conference on Computing Methodologies and Communication,  988-993, 2019.

\bibitem{o}  O. Regev, Improved inapproximability of lattice and coding problems with preprocessing. IEEE Transactions on Information Theory, 50(9): 2031-2037, 2004.

\bibitem{b}  BJ. Shi, The spectral norms of geometric circulant matrices with the generalized  k-Horadam numbers.  Journal of Inequalities and Applications ,14, 2018.

\bibitem{d2} D. Stehle,  R. Steinfeld, Making NTRU as secure as worst-case problems over ideal lattices. In K.G.Paterson(eds) Advances in Cryptology, Lecture Notes in Computer Sciences, Springer Berlin Heidelberg, 6632: 27-47, 2011.

\bibitem{y} Y. Yasin, N. Taskara, On the inverse of circulant matrix via generalized k-Horadam numbers. Applied Mathematics and Computation , 223:191-196, 2013.

\bibitem{z}  ZY. Zheng, WL. Huang, J. Xu, K. Tian,  A Generalization of Cyclic Code and Applications to Public Key Cryptosystems.


\end{thebibliography}
\end{document}